\documentclass[11pt,twoside]{article}
\usepackage{amssymb,amsmath,epsfig,multicol}
\usepackage{graphicx}
\usepackage{subfigure}
\usepackage{color}
\usepackage{fancyhdr}
\usepackage{lastpage}
\usepackage{amsfonts}
\usepackage[perpage,symbol*]{footmisc}

\newtheorem{theorem}{Theorem}

\newtheorem{lemma}[theorem]{Lemma}

\newtheorem{proposition}[theorem]{Proposition}
\newtheorem{remark}[theorem]{Remark}

\newenvironment{proof}[1][Proof]{\noindent\textbf{#1.} }{\ \rule{0.5em}{0.5em}}

\evensidemargin 0pt \fancypagestyle{plain}
 {
\fancyhf{}

}
\begin{document}



\begin{center}
\Large\textbf{Similarity flow solutions of a non-Newtonian power-law
fluid
 }\\
\ \ \newline  \large Mohamed Guedda, Zakia Hammouch
 \footnote{\textbf{Corresponding author}. \indent \emph{E-mail address}:
zakia.hammouch@u-picardie.fr -- Tel: 00 3 33 22 82 78 41}

 \end{center}
\begin{center} {\it LAMFA, CNRS UMR 6140, Universit\'e de Picardie
Jules Verne\\
Facult\'e de Math\'ematiques et d'Informatique,
33, rue Saint-Leu 80039 Amiens, France}

\end{center}

\begin{quote}
\textbf{Abstract:} In this paper we present a mathematical analysis
for a steady-state laminar boundary layer flow, governed by the
Ostwald-de Wael power-law model of an incompressible non-Newtonian
fluid past a semi-infinite power-law stretched flat plate
 with uniform free stream velocity. A generalization of the usual Blasius
similarity transformation is used to find similarity solutions
\cite{blasius}. Under appropriate assumptions, partial differential
equations are transformed into an autonomous third-order nonlinear
degenerate ordinary differential equation with boundary conditions.
Using a shooting method, we establish the existence of an infinite
number of  global unbounded solutions. The asymptotic behavior is
also discussed. Some properties of those solutions depend on the
viscosity power-law index. 

\end{quote}

\section{Introduction}
In view of their wide applications in different industrial
processes, and also by the interesting mathematical features
presented their equations, boundary-layer flows of non-Newtonian
fluids have motivated researchers in many branches of engineering in
recent years. The most frequently used model in non-Newtonian fluid
mechanics is the Ostwald-de Wael model (with a power-law rheology
\cite{astarita,astin,bohme,denier,howell}), which the relationship
between the shear stress and the strain rate is given as
follows
\begin{equation}\label{1}
 \tau_{xy}=k|u_{y}|^{n-1} u_{y},
\end{equation}
for $n=1$ the fluid is called Newtonian with dynamic coefficient of
viscosity $k$. For $n>1$ the behavior of the fluid is dilatant or
shear-thickening and for $0<n<1$ the behavior is shear-thinning, in
these cases the fluid is non-Newtonian and $k$ is the fluid
consistency. In this work we shall restrict our study to the
dilatant fluids, then throughout all the paper, the exponent $ n $
will be taken in the range $(1,\infty).$
 \noindent The problem of laminar flows of power-law non-Newtonian fluids  have been studied
by several authors. For the sake of brevity, we mention here some
examples, Acrivos et al.\cite{acrivos} and Pakdemirli \cite{pakde}
derived the boundary layer equations of power-fluids, Mansutti and
Rajagobal \cite{mansutti} investigated the boundary layer flow of
dilatant fluids. Adopting the Crocco variable formulation, Nachman
and Talliafero \cite{nachman} established existence and uniqueness
of similarity solution for a mass transfer problem. Filipussi et
al. \cite{filip} obtained similarity solutions and their
properties using a phase-plane formalism. Recently numerical solutions have
been given by Ece and B\"{u}y\"{u}k in \cite{ece} for the
steady laminar free convection over a heated flat plate. \\More recently Guedda \cite{zamp} studied the free convection problem of a Newtonian fluid, he showed the existence of an infinite number of solution and studied their asymptotic behavior. In this work we aim to extend the analysis of \cite{zamp} to the non-Newtonian case, we are interested
also in the effect of the power-law index on the existence and the asymptotic
behavior of solutions.\\The remainder of this work is organized as follows, in the next
section, we introduce the mathematical formulation of the problem, section 3 deals
with some preliminary tools which will be useful in section 4 and 5 to
prove the main results. Finally, we give some concluding remarks in section 6.
\section{Similarity procedure}
The problem is geometrically defined by a semi-infinite  power-law stretched rigid
plate, over which flows a non-Newtonian fluid obeying to
(\ref{1}). The main hypotheses for the mathematical formulation
of this problem are given by:
   \begin{itemize}
    \item Two-dimensional, incompressible and steady-state laminar flow,
    \item Physical properties are taken as constants,
    \item Body force, external gradients pressure and viscous dissipation are neglected.
\end{itemize}
Under these assumptions, and referred to a Cartesian system of
coordinates $ Oxy,$ where $ y = 0 $ is the plate, the $x$-axis is
directed upwards to the plate and the $y$-axis is normal to it, the
continuity and momentum equations can be simplified, within the
range of validity of the Boussinesq approximation \cite{acrivos}, to the
following equations
\begin{equation}\label{2}
\left\{\begin{array}{l}
 u u_{x} + vu_{y}=\nu (|u_{y}|^{n-1}u_{y})_y,\\
 \\
 u_{x}+v_{y}=0,
\end{array}\right.
\end{equation}
\noindent The functions $u$ and $v$ are the velocity components
in the $x-$ and $y-$ directions respectively. \\The boundary
conditions accompanied equation\, (\ref{2}) are given by
\begin{equation}\label{3}
u(x,0) = U_{w}(x), \quad v(x,0) = V_{w}(x),\quad u(x,y) \to 0\quad
\mbox{ as } y \to \infty.
\end{equation}
 \noindent The functions :\, $ U_w(x) =  u_{w}x^m $ is called the stretching
velocity and $u_{w}>0$,  the exponent $ m $ is negative, and
$V_{w}(x)=v_{w}x^{\frac{m(2n-1)-n}{n+1}}$ is the suction/injection velocity where
$v_{w}>0$ for suction and $v_{w} < 0$ for injection.\\
\noindent From the incompressibility of the
fluid we introduce the dimensionless stream function
$\psi=\psi(x,y)$  satisfying\,($ u =\psi_{y},\,\, v = -\psi_{x}$).\\
Hence equations (\ref{2})  are reduced to the single equation
\begin{equation}\label{4}
\psi_{y}\psi_{xy} -
\psi_{x}\psi_{yy} = \nu (|\psi_{yy}|^{n-1}\psi_{yy})_y.
\end{equation}
The boundary conditions (\ref{3}) are transformed into
\begin{equation}\label{5}
\psi_{y}(x,0) = u_{w}x^m,\quad \psi_{x}(x,0) = - v_{w}x^{\frac{m(2n-1)-n}{n+1}},\quad
\psi_{y}(x,y)\to 0 \mbox{ as } y \to \infty.
\end{equation}
Since the  broad goal of this paper is to obtain similarity
solutions to (\ref{4}),(\ref{5}) we introduce the following
similarity transformations
\begin{equation}\label{6}
  \psi(x,y) := Ax^\alpha f(t),\quad t := B\frac{y}{x^\beta}.
\end{equation}
\noindent Where $A,B,\alpha$ and $\beta$ are real numbers, $f$ is
the transformed dimensionless stream function and $t$ is the
similarity variable. In terms of (\ref{6}), equation (\ref{4})
can satisfy the ordinary differential equation of the shape function :
$f $\begin{equation}\label{7}
(|f''|^{n-1}f'')' +\alpha ff''=(\alpha-\beta){f'}^2,
\end{equation}
where the primes denote differentiation with respect to $t,$ if
and only if the following
\begin{center}
 $\alpha(2-n)+\beta(2n-1) = 1,\,\, \mbox{and} \quad \alpha-\beta=m,$
\end{center}
holds, and the parameters $A,B$ and $\nu$ satisfy
\begin{equation}\label{8}
\nu A^{n-2}B^{2n-1}=1, \quad \mbox{and} \quad m=\alpha-\beta.
\end{equation}
\noindent It follows that
\begin{center}
$\displaystyle \alpha= \frac{1+m(2n-1)}{n+1},\quad \beta= \frac{1+m(n-2)}{n+1}, \quad \mbox{and} \quad p=\frac{m(2n-1)-n}{n+1}.$
\end{center}
\noindent Consequently, we have
\begin{equation}\label{9}
  \psi(x,y) := \nu^{\frac{1}{n+1}} x ^{\frac{1+m(2n-1)}{n+1}}f(t),\quad t :=
 \nu^{-\frac{1}{n+1}}yx^{-\frac{1+m(n-2)}{n+1}}.
\end{equation}
The corresponding boundary conditions (\ref{5}) are expressed as
\begin{equation}\label{10}
f(0) = \frac{-v_w}{A\alpha},\quad f'(0) =\frac{u_w}{A} , \quad f'(\infty) =
\lim_{t\to\infty}f'(t) = 0.
\end{equation}
In the remainder, we deal with the following problem
\begin{equation}\label{11}
\left\{\begin{array}{l}
 (|f''|^{n-1}f'')'+\alpha ff''-m{f'}^2=0,\\
\\
f(0) = a,\quad f'(0) = b, \quad f'(\infty) =\lim_{t\rightarrow \infty}f'(t) = 0.
    \end{array}\right.
\end{equation}
For Newtonian fluids $(n=1)$, problem (\ref{11}) reads
\begin{equation}\label{12}
\left\{\begin{array}{l}
f''' + \alpha ff''- m {f'}^2= 0,\\
\\
f(0) = a,\quad f'(0) = b,\quad f'(\infty) = 0.
\end{array}\right.
\end{equation}
We notice that this problem
arises from two different contexts in fluid mechanics when looking
for similarity solutions. First, in natural convection along a
vertical heated flat plate, embedded in a saturated porous medium,
where the temperature is a power function with the exponent $m$, for
more details, we refer the reader to \cite{zamp,bbt,cheng} and
the references therein. Equation (\ref{12}) appears
also in the study of the boundary layer flow, of a Newtonian fluid,
adjacent to a stretching surface with a power-law velocity (see
\cite{bzaa,mpk}). In \cite{bbt}, the authors proved that
(\ref{12}) with $a=b=0$, has a  solution (which is bounded) if $
m \geq - \frac{1}{3} $ and this solution is unique for $ 0 \leq
m\leq \frac{1}{3}.$  In \cite{bzaa} the author gives a complete study about existence and nonexistence of solutions to (\ref{12}), where $b=1$.\\ Recently some new results have been obtained in
\cite{zamp}. The author considered the problem (\ref{12}), where $  m \in (-\alpha,0).$ He showed that, under some assumptions, problem  (\ref{12}) has  an infinite number of unbounded solutions and
these solutions satisfy $ f(t) \sim t^{\frac{\alpha}{\alpha-m}},$  as $ t $ goes to infinity.\\ Based on the results of \cite{zamp}, the interest in this work will be in existence and asymptotic behavior of solutions of problem (\ref{11}).

\section{Preliminary results}
As it is announced above, the existence of solutions will be
established by a shooting method. We replace the boundary condition
at infinity by $f''(0) = d, $ where $ d\not=0. $ Therefore, we consider the initial value problem
\begin{equation}\label{ivp}
\left\{\begin{array}{l}
 (|f''|^{n-1}f'')'+\alpha ff''-m{f'}^2=0,\\
\\
f(0) = a,\quad f'(0) = b,\quad f''(0) = d.
\end{array}\right.
\end{equation}

\noindent We shall see that for appropriate $ d $ problem (\ref{ivp}) has a global unbounded
solution and this solution satisfies the boundary condition at
infinity.
\begin{remark}
{\rm We notice that for $ n \not=1,$ equation
(\ref{11})$_1$ can be  degenerate or singular at the point $
t_0 $ where   $ f''(t_0) = 0.$ The existence the $ t_0
$ is done for  $ d > 0. $  We shall see also that $ f'''$
is not bounded at $t_0$ (the solution $ f $  is then not
classical). By a solution to (\ref{11}) we
will mean a function $ f \in C^2(0,\infty) $ such that $ |f''|^{n-1}f''\in C^1(0,\infty), f'(\infty) = 0$ and $f''(\infty)=0$. Note also that  any solution is classical
on any interval where  the second derivative does not change the sign.}
\end{remark}
 Consider now the initial value problem (\ref{ivp})
with $ n > 1,   a, d \in \mathbb{R}, b \geq 0 $ and $m\in (-\alpha,0).$ \\By the classical theory of ordinary differential equations the above problem has  local (maximal) solutions on some interval $(0,T_d), T_d \leq \infty $ and they are uniquely determined by $ d $ ($ d\not=0).$ Let us denote this such solution by $ f_d.$  Integrating (\ref{ivp})$_1$ to het the following identity
\begin{equation}\label{13}
|f''_d|^{n-1}f''_{d}(t) +\alpha
f'_{d}(t)f_d(t) = |d|^{n-1}d +\alpha a
b+(m+\alpha)\int_0^t{f'_d}(s)^2 ds ,\quad \forall \ t <T_d.
\end{equation}
which will be used later for proving some results.\\
A solution $f_d$ of (\ref{11}),  is of class $ C^2$ on $[0,T_d),$ and satisfies $ |f''_{d}|^{n-1}f''_d\in C^1([0,T_d)).$
We shall  investigate whether $ f_d $ admits an entire extension.  First, we give the following result characterizing the existence time $T_d$.
\begin{proposition} \label{cop} Let $ f_d $ be the local solution to (\ref{ivp}), if $T_d$ is finite then the functions $f_d$, $f'_d$ and $f''_d$ are unbounded as $t$ approaches $T_d$ from below.
\end{proposition}
\begin{proof}{\rm
Similar to \cite{bzaa,coppel}.}
\end{proof}\\

\noindent Let us note also that if we require a classical solution of (\ref{ivp}) (ie. $ f \in C^{3}(0,\infty)),$ it
is possible that $ f $ ceases to exist at some $ T < \infty $ and such that $ f, f'$ and $ f'' $ remain bounded on $
[0,T).$ More precisely we have the following result.
\begin{proposition} \label{p1} Let $ f_d $ be the local solution to (\ref{ivp}) where $n>1$ ans $d\not=0.$ Assume that there exists $ t_0 \in (0,T_d) $ such that $
f''_{d}(t_0) = 0. $ Then $d > 0,  f''_d < 0 $ on $(t_0,T_d) $
and $f'''_d$ is unbounded on $ (0,t_0) $.
\end{proposition}
\begin{proof} {\rm Assume first that $ d < 0.$ Therefore
$f''_d< 0 $ on $[0,\varepsilon), \varepsilon $ small, and the
following equation
\begin{equation}\label{14}
n|f''_d|^{n-1}f''_d+ \alpha ff'' - mf'^2 = 0,
\end{equation} holds
on $ (0,\varepsilon).$ Hence
\begin{equation}\label{15}
(f''_{d} e^{F})'= \frac{m}{n}e^{F}|f''_d|^{1-n}{f'_d}^2,\quad \mbox{on}\quad (0,\varepsilon),
\end{equation}
where
\[ F(t) = \frac{\alpha}{n}\int_0^t  f_d |f''_d|^{1-n}(s)ds.\] Consequently,  the function $ t\to  f''_d
e^{F}(t) $ is decreasing, and then $ f''_{d}(t)$ remains negative for all $ t \in
[0,T_d).$ A contradiction. Then $ d > 0.$ Actually, we have $f''_d> 0, f'_d> b $ on $ (0,t_0) $ and equation
(\ref{14}) holds on $(0,t_0).$ Assume now that $ f'''_d$
is bounded on $ (0,t_0).$ Thanks to equation (\ref{11})$_1$ we deduce that $
f'_d(t_0) = 0$ this is contradiction with $f'(0)> b$.}
\end{proof}\\

\noindent From the above we can see, in particular, that $ f''_d < 0 $ on $
(0,T_d) $ for any $ d < 0. $ Then $ f_d \in C^\infty([0,T_d)).$ While for the
case $ d > 0 $ the solution $ f_d $ is not classical.
\begin{proposition}\label{p2}
Let $ f_d $ be the local solution to (\ref{ivp}) for $d\neq 0$ and $n>1$.
If $ T_d < \infty $ then $\lim_{t\rightarrow T_d}  f_d(t) = -\infty.$
\end{proposition}
\begin{proof}{\rm  First we show that $ \sup_{[0,T_d)}|f_d(t)|= \infty.$  Suppose not and $ f''_{d}(t_0)= 0 $ holds,
for some $ t_0 \in (0,T_d). $  From (\ref{13}) we get
\[-( -f''_d)^{n}(t) +
\alpha f'_{d}(t)f_d(t) = \alpha f'_{d}(t_0)f_d(t_0) +
(m+\alpha)\int_{t_0}^t{f'_{d}(s)}^2ds,\quad \forall\ t_0 < t <
T_d.\] Hence
\[\frac{\alpha}{2}f_d^2(t)-\alpha f'_{d}(t_0)f_d(t_0)(t-t_0)
- \frac{\alpha}{2}f_d^2(t_0) = (m+\alpha)\int_{t_0}^t\int_{t_0}^\tau
{f'_d}^2(s)ds d\tau + \int_{t_0}^t(-f''_d)^{n}(s)ds.
\]
Since the right-hand side of the above  is  monotonic increasing with respect to $t$, the function $ f_d $ has a finite limit as $ t\to T_d.$ Consequently the function  $ (-f''_d)^{n} $ is integrable on $ (t_0,T_d)$.
Since $ n > 1 $ we deduce that $ f''_d$ is also
integrable on $ (t_0,T_d).$ Therefore  the function $ f'_d $
is bounded. Next we use  (\ref{13}) to deduce that $ f''_d $
is also bounded. A contradiction with Proposition \ref{cop}. \\It remains to prove that the hypothesis
$ f''_d> 0 $ on $ (0,T_d) $ leads also to a contradiction. Actually, in such situation, we know that $ f_d $ is classical and satisfies (\ref{14}), which yields to
\[(f''_d)^{n-2}f'''_d \leq -\frac{\alpha}{n}f_d, \]
and then
\[ (f''_d)^{n-2}f'''_d\leq \frac{\alpha}{n}\sup_{[0,T_d)}|f_d(t)|.\] Therefore $
f''_d $  and  $ f'_d$ are bounded. A contradiction.\\
Because$ f_d $ is monotonic on $(\tau, T_d), $ for some $ 0
<\tau < T_d, $ we deduce that $ \vert f_d(t) \vert $ goes to infinity as $ t\to T_d.$ Finally, to show that $ f_d(t)\to -\infty $
as $ t\to T_d$ we assume on the contrary that  $ \lim_{t\to T_d}f_d(t)=\infty.$ Hence the functions $ f_d $ and
 $ f'_d $  are
positive on $ (\tau,T_d).$ Moreover, using
(\ref{11})$_1,$ we can deduce from the Energy-function defined by
\begin{equation}\label{16}
E(t) = \frac{n}{n+1}|f''_d(t)|^{n+1} -
\frac{m}{3}{f'_{d}}^3,
\end{equation}
and satisfies $ E'(t) = -\alpha f_d {f''_d}^2 \leq 0.$ That $ f''_d $ and $ f'_d$  are bounded. Hence $ f_d $ is also bounded and this is a contradiction with Proposition.\ref{cop}. We conclude that if $T_d$ is finite the function $f_d$ goes to minus infinity as $t\rightarrow  \infty$.}
\end{proof}

\section{Existence of solutions}
In this section we shall obtain a
sufficient condition on $d$ such that the local solution $f_d$
of  (\ref{ivp}) is global and satisfies the condition $f'_d(\infty)= 0.$   We show that, for each $ d $ satisfying $ |d|^{n-1}d
>  -\alpha ab,\ f_d $ exists on the entire positive axis $\mathbb{R}^+$  and satisfies $
f'_d(\infty)= 0.$  We begin by a simple observation that: if $
m +\alpha > 0 $ and $|d|^{n-1}d > -\alpha ab,$
(\ref{13}) yields the important fact that $ f_d $ cannot have a
local maximum. Thus we prove the following result.
\begin{theorem}\label{th1} Let $ a \in \mathbb{R}, b\geq 0 $ and $
m \in (-\alpha,0).$ For any $ d $ such that $|d|^{n-1}d >
-\alpha ab,$
  there exists a unique global solution $ f_d, $ to
(\ref{ivp}), which goes to infinity with $ t, $
and its first and second derivative  tend to $ 0 $ as
$ t $ approaches infinity.
\end{theorem}
For our analysis, we need to distinguish two cases for the parameter $ a = f_d(0);$ namely $ a \geq 0 $ and $ a < 0. $ First we
prove the following lemma.
\begin{lemma}\label{l1} If $ a \geq 0 $ and $|d|^{n-1}d > - \alpha ab\
$ the functions $ f'_d $ and  $ f_d  $ are positive on $
(0,T_d)$  and $ T_d = \infty; $  that is $ f_d $ is global.
Moreover $f'_d$ and $f''_d$ are bounded.
\end{lemma}
\begin{proof}{\rm Because $ \vert d\vert^{n-1}d + \alpha ab> 0, $  the
first assertion of the lemma is immediate  from (\ref{13}). To
demonstrate that $ T_d = \infty $ it suffices to  show that $ f_d $
remains bounded on any bounded interval $[0,T]$. Let us
consider the Lyapunov function $ E $
for $ f_d $ defined by (\ref{16}). Since
\[ E'(t) = -\alpha f_d{f''_d}^2\leq
0,\] thanks to (\ref{11})$_1$, it is seen that
\[ \frac{n}{n+1}{\vert f''_{d}(t)\vert}^{n+1} -
\frac{m}{3}{f'_d(t)}^3 \leq \frac{n}{n+1}\vert d\vert^{n+1}
-\frac{m}{3}{b}^3,\quad \forall t < T_d.\] This  in turn implies
that $ f''_d, f'_d $ and then $f_d $ are bounded on $ [0,T].$}
\end{proof}

\begin{lemma}\label{l2} If $ a \geq 0 $ and $\vert d\vert^{n-1}d > -\alpha ab,
f_d(t) $ tends to infinity with $ t, f'_d $ and $
f''_d$  tend to zero as $ t\to \infty. $
\end{lemma}
\begin{proof}{\rm Since  $ f_d^\prime
$ is monotonic on $ (t_1,\infty),\ t_1 ,$ large enough, and bounded
  there exists a $ l \geq 0$
such that
\[  \lim_{t\to \infty} f_d^\prime(t) = l.\]
This implies the existence of a sequence $ (t_n) $ tending to infinity with $ n $ satisfying $ \lim_{n\to+\infty}
f_d^{\prime\prime}(t_n) = 0 $ and then \(\lim_{t\to\infty} f_d^{\prime\prime}(t) = 0,\)
with the help of the energy function $ E. $\\
Now we assume  that $ f_d $ is bounded, therefore $ l=0. $
Subsequently $$\displaystyle \vert d\vert^{n-1}d  +
\alpha ab+(m+\alpha)\int_0^{\infty}f_d^\prime(t)^2dt = 0.$$
 This is impossible. Therefore $ f_d $ is unbounded and then $
\lim_{t\to+\infty}f_d(t) = \infty.$ It remains to prove that $ l =
0.$ Assume on the contrary that $ l > 0.$ Together with (\ref{13}) we get
\[{|f''_{d}|}^{n-1} f''_d(t) = -\alpha l^2t +
(m+\alpha)l^2t + o(t),\]
\[{\vert f''_d\vert}^{n-1} f''_d(t) = m l^2t + o(t),\]
as
$ t $ approaches infinity, that
This is only possible if $ m = 0.$ Consequently $ l = 0.$}
\end{proof}

\noindent Next we  consider the case  $ a < 0. $\\
 The  first simple consequence is that
$ f_d(t) < 0 $ and $ f_d^\prime(t) > 0 $ for small $ t > 0. $ Since
$ f_d $ cannot have a local maximum, we have two possibilities\\
$\bullet$ \, Either $ f_d (t) $ vanishes at a some point and remains positive
after this point.\\
$\bullet$\,  Or $ f_d (t) < 0 $ for all $ t > 0. $ \\  Hence the
proof of Theorem \ref{th1} is completed by the following lemma.
\begin{lemma}\label{l3}
Assume $ a < 0 $ and $ \vert d\vert^{n-1}d > - \alpha ab. $ Then
$ f_d $ has exactly one zero, goes to $\infty$ with $t$, and the functions
$ f_d^\prime, f_d^{\prime\prime} $ converge to $ 0 $ as $ t \to
\infty.$
\end{lemma}
\begin{proof} {\rm
Assume that the first assertion holds. Since $ f^\prime_d $ is positive we deduce that $
f_d $ is bounded and then is global.  On the other hand, using
(\ref{13})
  one sees that $ f^{\prime\prime}_d > 0. $  Therefore we get
\(  \lim_{t\to \infty}f_d(t) \in (a,0] \) and \(\lim_{t\to
\infty}f_d^\prime(t)=0,\) since $ f^\prime_d $ is monotonic. This is
absurd since $ f_d^\prime $ is positive and increasing function.
Hence $ f_d $ has exactly one zero, say $ t_0. $ To finish the Proof
of Lemma \ref{l3} and therewith that of Theorem \ref{th1} we note
that the new function \[ h(t) = f_d(t+t_0) \] satisfies equation
(\ref{11})$_1$  and  $$ h(0) \geq 0,\quad h^{\prime\prime}(0) > -\alpha
h(0)h^\prime(0).$$ Therefore we use Lemmas \ref{l1} and \ref{l2}
to conclude.}
\end{proof}\\
In the next result we complete our analysis on the existence of
global solutions by the case  $ b < 0.$
\begin{theorem}\label{th2} Let $ b < 0, a > 0 $ and $m \in (-\alpha,0).$ For any $
d > 0 $ satisfying
\begin{equation}\label{17}
  ad^{n} - \frac{1}{2}b^2d^{n-1} + \alpha a^{2}b > 0.
   \end{equation}
  The unique local solution,
$ f_d $ to (\ref{ivp}) is global unbounded and satisfies $
\lim_{t\to\infty}f_d^\prime(t)=\lim_{t\to\infty}f_d^{\prime\prime}(t)
= 0. $
\end{theorem}
\begin{proof}{\rm
Since $  a, d > 0 $ and $ b  < 0,  $ there exists a real $ t_0 > 0 $
such that $ f_d $ is positive, decreasing and convex on $(0,t_0).$
Define
\[ T = \sup\left\{t: f_d(s) > 0, f_d^\prime(s) < 0,
f''_d(s) > 0, \mbox{ for all } s \in (0,t)\right\}.\]
The real number $ T $  is larger than $ t_0 $ and may be infinite. \\
Assume that $ T = \infty.$  Then  the function $ f_d $ has a finite
limit at infinity and $ f_d^\prime(t) $ (and $f''_d$) go to zero as $ t \to \infty.
$ Since the
function
\[ H = f_d \vert f_d''\vert^{n-1}f''_d-
\frac{1}{2}{f'_d}^2\vert f''_d\vert^{n-1} + \alpha
f_d^{2}f_d^\prime, \] satisfies
\[ H^\prime = f(f_d^\prime)^2\left[ m + 2\alpha  +
\frac{\alpha(n-1)}{2n}\right] -
\frac{m(n-1)}{2n}(f_d^\prime)^4(f''_d)^{-1},\] thanks to
(\ref{11})$_1$, we deduce that $ H $ is increasing on $
(0,\infty).$ Hence for $t > 0$ we have
\[ H(0) < \lim_{t\rightarrow \infty} H(t) = 0, \]
which yields to
\[ad^{n} - \frac{1}{2}b^2d^{n-1} + \alpha a^2b < 0.\]
A contradiction. Therefore $T$ is infinite. Next, we assume that $f_d(T) = 0 $ or $ f''_d(T) = 0.$ Arguing as above we deduce $
H(0)  < 0 $ and then we get  a contradiction.  In conclusion if
condition (\ref{17}) holds the  function $ f_d $ has a local
positive minimum at some $ t_1 > 0.$ We use Theorem \ref{th1}
 to
deduce that the
  new function $ h(t) = f_d(t+t_1) $ is global, unbounded and
satisfies $ h^\prime(\infty)=h''(\infty)=0.$ The proof is
finished.}
\end{proof}\\
\begin{remark}{\rm
We notice that we can extend the results of Theorem \ref{th1} to the case $(-2\alpha ,-\alpha )$ by using the function $H$ defined above, as in the work by Guedda \cite{zamp} for the Newtonian case.
}
\end{remark}
\section{Asymptotic behavior}
In this section we shall derive the asymptotic behavior of any
possible global unbounded solution to (\ref{11})
for $ m \in (-2\alpha,0).$
  First we give the following result
\begin{lemma}\label{l4}
Let $f$ be a positive solution to (\ref{11}) for $m\in(-2\alpha,0)$. Then $f'$ goes to zero at infinity and $f''$ is negative.
\end{lemma}
\begin{proof}{
Since $f$ is monotonic on $[t_0,\infty)$, $t_0$ large enough, we get the positivity of $f'$ and $f$ on
$(t_{0},\infty)$. In addition  we use the Lyapunov function to get the boundedness of $f'$ and $f''$. Arguing as in the previous section we get that $f'\rightarrow 0$ and $f''<0$ for large $t$. }
\end{proof}
\begin{proposition}\label{p} Assume that $ n > 1$ and  $m\in(-2\alpha,0)$. Let $ f $ be a positive solution to
(\ref{11}). Then
\[ \lim_{t\to \infty} f_d(t)f_d^{\prime\prime}(t) =  \lim_{t\to
\infty} \left(\vert f''\vert^{n-1} f''\right)^\prime(t) = 0.
\]
\end{proposition}

\begin{proof}{\rm
Thanks to lemma (\ref{l4}) we have  $ f^\prime(t) > 0, f''(t) < 0 $ for all
$ t > t_0,\  t_0 $ large enough and $ f^\prime $  and $
f^{\prime\prime} $ tend to 0 as $ t \to \infty.$ Then equation (\ref{11})$_1$ can be
written as
\[ f''' + \frac{\alpha}{n} ff''\vert f''\vert^{1-n} =
\frac{m}{n}{f^{\prime}}^{2}\vert f''\vert^{1-n},\quad \forall t
> t_0.\]
 By differentiation we have
\begin{equation}\label{18}
f^{(iv)} + f^{\prime\prime\prime}\left[\frac{\alpha(2-n)}{n}\vert
f''\vert^{1-n}f-\frac{m(1-n)}{n}\vert
f''\vert^{-n-1}f''{f^\prime}^2\right] =
-\frac{\alpha-2m}{n}f^\prime f^{\prime\prime}|f''|^{1-n}.
\end{equation}
Then the function $ f^{\prime\prime\prime}e^{G} $ is monotonic
increasing on  $ (t_0,\infty), $ where
\[ G^\prime =\frac{\alpha(2-n)}{n}\vert
f''\vert^{1-n}f-\frac{m(1-n)}{n}\vert
f''\vert^{-n-1}f''{f^\prime}^2. \] This indicates that the function $ f^{\prime\prime\prime} $ has at most
one zero. Because $ f''$ is negative and goes to 0 at infinity,
we deduce that  $ f'''(t) > 0 $ on $ (t_1,\infty),$ for  $t_1 $
large. On the other hand,  from (\ref{11})$_1$ we deduce
\[ (\vert f''\vert^{n-1}f'')''  +(\alpha-2m)f^\prime f''= -\alpha f f'''.\] Therefore the function
$t \longmapsto  (\vert f''\vert^{n-1}f'')^\prime  +\frac{\alpha-2m}{2}{f^\prime}^2$
is positive and monotonic decreasing on
$(\inf\left\{t_0,t_1\right\},\infty). $ Together with the fact that
$ f^\prime $ tends to 0 as $ t\to \infty $ we deduce that
\[ \lim_{t\to+\infty} \left(\vert f''\vert^{n-1}
f''\right)^\prime(t)=0\] and then  we conclude that $ f f''(t) \to 0 $ as $
t\to\infty,$ thanks to (\ref{11})$_1$.}
\end{proof}
\begin{proposition}\label{pp}  Let $ f $ be a solution to (\ref{11}) where
$ n>1,$ $m \in (-2\alpha,0).$ Then
\[ \lim_{t\to+\infty} ff''(t) = \left\{\begin{array}{ll}
\infty,&\quad \mbox{ if }\,  m + \alpha > 0,\\
\\
L \in (0,\infty),&\quad \mbox{ if }\,   m +\alpha = 0,\\
\\
0, &\quad \mbox{ if }\,   m +\alpha < 0.
\end{array}\right.\]
\end{proposition}
\begin{proof}{\rm
Let $ f $ be a global solution to (\ref{11}). First
we claim that there exists $ t_0 \geq 0 $ such that
\[\vert f''\vert^{n-1}f''(t_0) + \alpha f f'(t_0) > 0. \]
Suppose not; that is
\[\vert f''\vert^{n-1}f''(t) + \alpha f f^\prime(t) \leq 0,\]
for all $ t \geq 0.$  Since $ f''(t) \to 0 $ as $ t \to \infty $
the following
\[ f''+ \alpha ff^\prime \leq 0\]
holds on some $ (t_1,\infty), t_1 $ large. Consequently the
function $ f^\prime + \frac{\alpha}{2}f^2 $ is decreasing and goes
to infinity with $ t,$ which is absurd. Now we use  the identity
\[\vert f''\vert^{n-1}f''(t) + \alpha f f^\prime(t) = \vert
f''\vert^{n-1}f''(t_0) + \alpha f f^\prime(t_0) +
(m+\alpha)\int_{t_0}^t {f^\prime}^2(s)ds,\] to deduce that $ f
f^\prime $ has a limit $ L \in [0,\infty] $  at infinity. This
limit is finite for $ \alpha + m =0. $ Assume that $ \alpha + m \neq 0.$ If $L\in (0,\infty)$ we get immediately that $f' \sim \sqrt{\frac{L}{\eta}}$ at infinity which implies that $ff'\rightarrow \infty$. A contradiction. Consequently $L\in \left\{ 0,\infty\right\}$, we use again the above identity to conclude that $ L = \infty $ if $ m+\alpha > 0 $ and $L=0$ if $ m+\alpha < 0 $.}
\end{proof}\\
\begin{remark}{\rm
We stress that the condition
\[\vert f''\vert^{n-1}f''(t_0) + \alpha f f'(t_0) > 0,\quad
f'(t_0) \geq 0  \] is necessary and sufficient to obtain a
global solution converging to plus infinity with $ t $ in the case
$ m \in (-\alpha,0).$ \\

 Now we are ready to give the result concerning the large $
t-$behavior of solutions to (\ref{11}).}
\end{remark}

\begin{theorem} \label{th3} Suppose  $ n>1, - 2\alpha < m  < 0. $  Let $ f $ be
a solution to (\ref{11}) such that \, $f\rightarrow \infty$. Then there exists a constant, $ A > 0, $ such that
\begin{equation}\label{19}
f(t) = t^{\frac{\alpha}{\alpha-m}}(A+ o(1)),
\end{equation}
as $ t \to \infty.$
\end{theorem}
\begin{proof}{\rm
Let $ f $ be a global solution to (\ref{11}). First we prove the result for the case $m+\alpha > 0$.\\
 Let $t_0 $ be a real number such that \( f'' < 0 \) and \( f''' > 0
$ on $(t_0,\infty).$ Dividing equation (\ref{11})$_1$ by $ff'$ gives
$$ \frac{(\vert f''\vert^{n-1}f'')'}{ff'}=m \frac{f'}{f}-\alpha \frac{f''}{f'}.$$
Integrating over $(t_1,t)$, for $t_1 > t_0$, gives
$$ \displaystyle {\int_{t_1}}^{t} \frac{(\vert f''\vert^{n-1}f'')'}{ff'} ds = \log{(f^{m}(t)f'^{-\alpha }(t)}
- \log{(f^{m}(t_1)-f'^{-\alpha }(t_1))}.$$
According to Proposition \ref{pp}, $ff'$ goes to infinity with $t$ and then the left hand side of the above is integrable, consequently $f^{m}f'^{-\alpha }$ has a positive finite limit at infinity. The desired asymptotic behavior (\ref{19}) follows by a simple integration.\\
Now we deal with the case $m+\alpha < 0$.
For this sake we define
\[ \Psi = \varphi(f)\vert
f''\vert^{n-1}f''-\frac{1}{2}\varphi^\prime(f)(f^\prime)^2\vert
f''\vert^{n-1}+\alpha \varphi(f)ff^\prime,\] where $ \varphi $ is a
smooth function. Then if follows from (\ref{11})$_1$
\[ \Psi^\prime(t) = {f^\prime}^2\left[ \alpha f\varphi^\prime(f) +
(\alpha+m)\varphi\right]-\frac{1}{2}\varphi''(f){f^\prime}^3\vert
f''\vert^{n-1} - \frac{n-1}{2}\varphi^\prime(f){f^\prime}^2\vert
f''\vert^{n-3}f''f'''. \] Let the function $\varphi$ be defined by
\[ \varphi(s) = s^{-\frac{m+\alpha}{\alpha}}.\]
It satisfies the following differential equation
\[ \alpha s\varphi^\prime(s) + (\alpha+m)\varphi=0.\]
 This implies that
\[ \Psi^\prime = -\frac{1}{2}\varphi''(f){f^\prime}^3\vert
f''\vert^{n-1} - \frac{n-1}{2}\varphi^\prime(f){f^\prime}^2\vert
f''\vert^{n-3}f''f''' \geq 0,\]
\[ \Psi = \varphi(f)\left[\vert
f''\vert^{n-1}f''-\frac{\alpha+m}{2\alpha f}(f^\prime)^2\vert
f''\vert^{n-1}+\alpha ff^\prime\right],\] and then
\[  \vert \Psi^\prime(t)\vert  \leq
\varepsilon\left[f^{-\frac{3\alpha+m}{\alpha}}f^\prime
+(n-1)(-f'')^{n-2}f'''\right],\] for all $ t \geq t_0,
t_0 $ large. Therefore $ \Psi^\prime $ is integrable on $[0,\infty)$ and then $ \Psi
$ has a finite limit at infinity, say $ L.$ Next we show
that $ L > 0.$ It will be sufficient to show that $ \Psi(t_1) > 0 $
for some $ t_1 $ large. Suppose not; that is for any $ t > t_2,\, t_{2}$ large we have
\[\vert f''\vert^{n-1}f''-\frac{\alpha+m}{2\alpha f}{f^\prime}^2\vert
f''\vert^{n-1}+\alpha ff^\prime \leq 0.\] Since $\alpha+m < 0$ then  $$f''+\alpha  ff' \leq 0,$$
 from which we deduce, as above, that $f'+\frac{\alpha}{2} f^{2}$ is  a decreasing  function going to infinity with $t$. A contradiction. We conclude that $\lim_{t\to\infty}f^{-\frac{m}{\alpha}}f^\prime =
\frac{L}{\alpha}.$ Finally,  a simple integration leads to estimate  (\ref{19}).\\
To finish, we pay attention to the case $m=-\alpha $. Here the identity (\ref{13}) leads to
$$ |f''|^{n-1}f'' + \alpha ff'= |\gamma|^{n-1}\gamma + \alpha ab.$$
According to Theorem \ref{th1}, $f$ is global and satisfies $f \sim t^{\frac{1}{2}}$ at infinity.
}
\end{proof}\\
\section{Conclusion}
The laminar two-dimensional steady boundary layer flow, of a
non-Newtonian incompressible fluid, over a stretching surface have
been considered. Using the shooting method, existence of
 global unbounded similarity  solutions have been shown,  the
dependency of those solutions on the power-law index has been
investigated, and their asymptotic behavior was also discussed.\\
Coming back  to the original problem (\ref{2}),(\ref{3})  we
find that, for $-2\alpha < m <0,$ the stream function
satisfies
$$\displaystyle \psi(x,y) \sim y^{\frac{1+m(2n-1)}{1+m(n-2)}} \quad \mbox{as} \quad yx^{\frac{(2-n)m-1}{n+1}}\to \infty.$$

\section*{Acknowledgments}
 The authors are indebted to the anonymous referees for valuable comments and
 Prof. Robert Kersner  for stimulating discussions. This work is an outcome of a PhD thesis \cite{thesis}, it was partially supported by la Direction des Affaires Internationales
(UPJV) Amiens France and by the Research Program: PAI No MA/05/116 (France-Morocco Scientific Cooperation Volubilis).

\end{document}